\documentclass[12pt,reqno]{article}

\usepackage[usenames]{color}
\usepackage{amssymb}
\usepackage{amsmath}
\usepackage{amsthm}
\usepackage{amsfonts}
\usepackage{amscd}
\usepackage{graphicx}

\newcommand{\seqnum}[1]{\href{https://oeis.org/#1}{\rm \underline{#1}}}

\usepackage[colorlinks=true,
linkcolor=webgreen,
filecolor=webbrown,
citecolor=webgreen]{hyperref}

\definecolor{webgreen}{rgb}{0,.5,0}
\definecolor{webbrown}{rgb}{.6,0,0}

\usepackage{color}
\usepackage{fullpage}
\usepackage{float}

\usepackage{graphics}
\usepackage{latexsym}
\usepackage{epsf}
\usepackage{breakurl}

\DeclareMathOperator{\FO}{FO}
\DeclareMathOperator{\per}{per}
\DeclareMathOperator{\ce}{ce}
\def\Enn{\mathbb{N}}

\begin{document}

\theoremstyle{plain}
\newtheorem{theorem}{Theorem}
\newtheorem{corollary}[theorem]{Corollary}
\newtheorem{lemma}[theorem]{Lemma}
\newtheorem{proposition}[theorem]{Proposition}

\theoremstyle{definition}
\newtheorem{definition}[theorem]{Definition}
\newtheorem{example}[theorem]{Example}
\newtheorem{conjecture}[theorem]{Conjecture}

\theoremstyle{remark}
\newtheorem{remark}[theorem]{Remark}

\title{The First-Order Theory of Binary Overlap-Free Words is Decidable}

\author{Luke Schaeffer\\
Institute for Quantum Computing (IQC) \\
University of Waterloo \\
Waterloo, ON N2L 3G1 \\
Canada\\
\href{lrschaeffer@gmail.com}{\tt lrschaeffer@gmail.com} \\
\ \\
Jeffrey Shallit\\
School of Computer Science \\
University of Waterloo \\
Waterloo, ON N2L 3G1 \\
Canada\\
\href{shallit@uwaterloo.ca}{\tt shallit@uwaterloo.ca}
}

\maketitle

\begin{abstract}
We show that the first-order logical theory of the binary overlap-free
words (and, more generally, the $\alpha$-free words for
rational $\alpha$, $2<\alpha \leq 7/3$), is decidable.
As a consequence, many results previously obtained
about this class through tedious case-based proofs
can now be proved ``automatically'', using a decision 
procedure.     
\end{abstract}

\section{Introduction}

Let $V_k(n)$ be the highest power of $k$ dividing $n$; thus, for 
example, we have $V_2 (48) = 16$.
A famous theorem of B\"uchi \cite{Buchi:1960}, as corrected and clarified by
Bruy\`ere et al.~\cite{Bruyere&Hansel&Michaux&Villemaire:1994},
states that for each
integer $k\geq 2$, the first-order logical theory $\FO(\Enn, +, <, 0, 1, V_k)$
is decidable.  (The logical structure $(\Enn, +, <, 0, 1, V_2)$
is sometimes called {\it B\"uchi arithmetic}; it is an extension of the
more familiar Presburger arithmetic.)
As a consequence, it follows that the first-order theory of
$k$-automatic sequences is decidable.

Recently decidability results have been proved for a number of interesting
infinite classes of infinite sequences.   For the paperfolding sequences,
see \cite{Goc&Mousavi&Schaeffer&Shallit:2015}.  For a class of Toeplitz words,
see \cite{Fici&Shallit:2022}.  For the Sturmian sequences, see 
\cite{Hieryonymi&Ma&Oei&Schaeffer&Schulz&Shallit:2022}.

In this paper we prove that a similar result holds for the
first-order theory of the binary overlap-free words (and, more generally,
for $\alpha$-power-free words for $\alpha$ a rational number
with $2<\alpha\leq 7/3$).    This allows
us to prove (in principle), purely mechanically, assertions about
the factors of such words, compare different overlap-free words, and
quantify over all overlap-free words or appropriate subsets of them.

A version of this
decision algorithm has been implemented using {\tt Walnut}, a theorem-prover
originally designed by Hamoon Mousavi \cite{Mousavi:2016,Shallit:2022},
and we have
used it to reprove various known results about overlap-free words, and
some new ones.  

\section{Definitions and basic concepts}

Let $x = e_{t-1} \cdots e_1 e_0$
be a word over the alphabet $\{0,1,2\}$.  We define
$[x]_2 = \sum_{0 \leq i < t} e_i 2^i$, the value of $x$
when interpreted in base $2$.   The case where the
$e_i \in \{0,1\}$ corresponds to the ordinary binary
representation of numbers; if the digit $2$ is also allowed,
we refer to the {\it extended binary representation}.
For example $[210]_2 = [1010]_2 = 10$.

Let $x = x[0..n-1]$ be a finite word of length $n$.  If $1 \leq p \leq n$
and $x[i]=x[i+p]$ for 
$0 \leq i < n-p$, then we say that $x$ has {\it period\/} $p$.
The least period is called {\it the\/} period, and is denoted
$\per(x)$.  The {\it exponent} of a nonempty word $x$ is defined
to be $|x|/\per(x)$.   
If $\exp(x) = \alpha$, we say that $x$ is an $\alpha$-power.
For example, the word {\tt onion} is a ${5 \over 2}$-power.

The supremum of $\exp(x)$, taken over all finite nonempty factors of $x$,
is called the {\it critical exponent} of $x$, and is denoted $\ce(x)$.
If $\ce(x) < \alpha$, we say that $x$ {\it avoids $\alpha$-powers\/} or
that $x$ is {\it $\alpha$-power-free}.
If $\ce(x) \leq \alpha$, we say that $x$ {\it avoids $\alpha^+$-powers} or
that $x$ is {\it $\alpha^+$-power-free}.   Thus when we talk about
power-freeness, we are using a sort of ``extended reals'', under
the agreement that $e < e^+ < f$ for all $e<f$; this very useful notational
convention was apparently introduced by Kobayashi \cite[p.~186]{Kobayashi:1986}.
These concepts extend seamlessly to infinite words.
A {\it square\/} is a $2$-power; an example in English is {\tt murmur}.
The {\it order\/} of a square $xx$ is defined to be $|x|$.

An {\it overlap\/} is a word of the form $axaxa$, where $a$ is a single
letter and $x$ is a possibly empty word.  For example, the French
word {\tt entente} is an overlap.   If a word has no factor that
is an overlap, we say it is {\it overlap-free}.  Equivalently, a word
is overlap-free iff it avoids $2^+$-powers.  Much of what we know
about overlap-free words is contained in Thue's seminal
1912 paper \cite{Thue:1912,Berstel:1995}.  For more recent
advances, see 
\cite{Fife:1980,Restivo&Salemi:1983,Restivo&Salemi:1985a,Carpi:1993a,Shur:2000,Rampersad:2007}.

The most famous infinite binary overlap-free word is 
$${\bf t} = 01101001\cdots, $$
the Thue-Morse sequence.
It satisfies the equation
${\bf t} = \mu({\bf t})$, as does its binary complement
$\overline{\bf t}$, where $\mu$ is the {\it Thue-Morse morphism}
mapping $0$ to $01$ and $1$ to $10$.  

We write $\mu^n$ for the $n$-fold composition of $\mu$ with itself.

A theorem of Restivo and Salemi \cite{Restivo&Salemi:1985a} provides
a structural description of finite and infinite binary overlap-free words
in terms of the Thue-Morse morphism $\mu$.
This result was extended to all powers $2 < \alpha \leq {7 \over 3}$ in
\cite{Karhumaki&Shallit:2004}, as follows:

\begin{theorem}
Let $S = \{ \epsilon, 0, 1, 00, 11 \}$.
Let $2 < \alpha \leq {7 \over 3}$ be a rational number $(p/q)$ 
or extended rational $(p/q)^+$.

\begin{itemize}
\item[(a)]
Suppose $w$ is a finite binary $\alpha$-free word.  
Then there exist words 
$$x_0, x_1, \ldots, x_n, y_0, y_1, \ldots, y_{n-1} \in S$$
such that $x = x_0 \mu(x_1) \mu^2(x_2) \cdots \mu^n(x_n) \mu^{n-1}(y_{n-1})
\cdots \mu(y_1) y_0$.

\item[(b)]
Suppose $\bf w$ is an infinite binary $\alpha$-free word.
Then there exist infinitely many words $x_0, x_1, \ldots$ such that
${\bf w} = x_0 \mu(x_1) \mu^2(x_2) \cdots$, or finitely many
words $x_0, x_1, \ldots, x_n$ such that
${\bf w} = x_0 \mu(x_1) \mu^2(x_2) \cdots \mu^n(x_n) {\bf t}$ or
${\bf w} = x_0 \mu(x_1) \mu^2(x_2) \cdots \mu^n(x_n) \overline{\bf t}$.
\end{itemize}
\label{restivo}
\end{theorem}

Of course, not all sequences of choices of the $x_i$ and $y_i$ result
in overlap-free (or $\alpha$-power-free) words.  For example,
taking $x_0 = 0$, $x_1 = 11$, $y_0 = 1$ gives the word
$0\, \mu(11) \, 1 = 0 \, 10 \, 10 \, 1 = (01)^3$.
See \cite{Carpi:1993a,Cassaigne:1993,Jungers&Protasov&Blondel:2009,Blondel&Cassaigne&Jungers:2009,Shallit:2011a,Rampersad&Shallit&Shur:2011} for more details.

\section{Decidability for binary overlap-free words}
\label{decision}

Theorem~\ref{restivo} is our basic tool.  We
code the words $x_i$ and $y_i$ with the following correspondence:
\begin{align*}
g(1) &= \epsilon  \\
g(2) &= 0  \\
g(3) &= 1 \\
g(4) &= 00 \\
g(5) &= 11  .
\end{align*}
The finite code $c_0 c_1 \cdots c_t \in \{1,2,3,4,5\}^*$
is understood to specify the finite {\it Restivo word\/}
$$R(c_0 c_1 \cdots c_t) 
= g(c_0) \mu(g(c_1)) \mu^2(g(c_2)) \cdots \mu^t(g(c_t))$$
and the infinite code $c_0 c_1 \cdots \in \{1,2,3,4,5\}^\omega$
is understood the specify the infinite Restivo word
$$R(c_0 c_1 \ldots) = g(c_0) \mu(g(c_1)) \mu^2(g(c_2)) \cdots.$$
Thus the Restivo words correspond to ``one-sided'' part (a)
of Theorem~\ref{restivo}.

Similarly, the finite codes $c_0 c_1 \cdots c_t,
d_0 d_1 \cdots d_{t-1} \in \{1,2,3,4,5\}^*$ are understood to specify
the finite {\it Salemi word\/}
$$ S(c_0 c_1 \cdots c_t, d_0 d_1 \cdots d_{t-1}) =
g(c_0) \mu(g(c_1)) \mu^2(g(c_2)) \cdots
\mu^t(g(c_t)) \mu^{t-1} (g(d_{t-1})) \cdots
\mu^1 (g(d_1)) g(d_0).$$
Thus, the Salemi words correspond to the ``two-sided'' part (b)
of Theorem~\ref{restivo}.

We emphasize that we do {\it not\/} require that Restivo words and Salemi
words be overlap-free, only that they are of the form given above
with the $c_i, d_i \in S$.

We prove the following results:
\begin{theorem}
Let $N_{c,d}$ be the structure
$(\Enn, <, +, 0, 1, n \rightarrow V_2(n), n \rightarrow S(c,d)[n] )$, where
we augment B\"uchi arithmetic by a finitely coded Salemi word $S(c,d)$.
Let $K_{\rm finite} = \{ N_{c,d} \, : \, c,d \in \{1,2,3,4,5\}^* \}$.
Then the first-order logical theory $\FO(K_{\rm finite})$ is
decidable.
\label{thm1}
\end{theorem}

\begin{theorem}
Let $N'_{\bf c}$ be the structure
$(\Enn, <, +, 0, 1, n \rightarrow V_2(n), n \rightarrow R({\bf c})[n] )$, where
we augment B\"uchi arithmetic by a Restivo word $R({\bf c})$ with infinite
code $\bf c$.   Let $K_{\rm infinite} = \{ N'_{\bf c} \, : \, {\bf c} \in
\{1,2,3,4,5\}^\omega \}$.
Then the first-order logical theory $\FO(K_{\rm infinite})$ is
decidable.
\label{thm2}
\end{theorem}

\begin{proof}[Proof of Theorems~\ref{thm1} and \ref{thm2}.]
The basic strategy of our decision procedure can be found in the
papers of B\"uchi \cite{Buchi:1960} and Bruy\`ere et 
al.~\cite{Bruyere&Hansel&Michaux&Villemaire:1994} mentioned
previously.  Since B\"uchi arithmetic itself is decidable,
and is powerful enough to express the computations of a deterministic
finite automaton (DFA) or deterministic finite automaton
with output (DFAO), it suffices to construct a 
DFAO
computing $n \rightarrow S(c,d)[n]$ and $n \rightarrow R({\bf c})[n]$.
Here the automata take the words coding $c, d, {\bf c}$ and $n$ (in binary) in
parallel, and compute the $n$'th bit of the corresponding word.
We call these the {\it lookup automata}.
For the Salemi words we use ordinary finite automata,
and for infinite binary words we use B\"uchi automata.

We construct the lookup automata in stages. First we describe
how to compute the lookup automaton
for the finite Restivo word
$$R(c_0 c_1 \cdots c_t) 
= g(c_0) \mu(g(c_1)) \mu^2(g(c_2)) \cdots \mu^t(g(c_t)) .$$

Given $n$, our first task is to determine in which factor the index
$n$ lies.
To achieve this, we observe that $|\mu^j (g(i))| = 2^j |g(i)| = a \cdot 2^i$
for $i \in \{1,2,3,4,5\}, a \in \{0,1,2\}$.   Defining the morphism $h$ as follows:
\begin{align*}
h(1) &= 0 \\
h(2) &=1 \\
h(3) &= 1\\
h(4) &= 2\\
h(5) &= 2,
\end{align*}
we see that $h(i) = |g(i)|$.
If we now
interpret $h(c_t \cdots c_1 c_0)$ as a generalized base-$2$ number
with the digit set $\{0,1,2\}$, we see that the $n$'th symbol of
$R(c_0 c_1 \cdots c_t) $ 
is equal to the $n-k$'th symbol of $\mu^i(g(c_i))$,
where
\begin{equation}
[h(c_{i-1} \cdots c_1 c_0)]_2 \leq n < [h(c_i \cdots c_1 c_0)]_2 ,
\label{ineq}
\end{equation}
and $k = [h(c_{i-1} \cdots c_1 c_0)]_2$.   Here all words are indexed
starting at position $0$.    We can find the appropriate
$i$ with an existential quantifier that checks the inequalities
\eqref{ineq}.  

Since $n$ is given in binary, we need 
a normalizer that takes as input two strings in parallel, one over
the larger digit set $\{0,1,2\}$ and one over the ordinary digit
set $\{0,1\}$, and accepts if they represent the same number when
considered in base $2$.  This is done with the automaton in 
Figure~\ref{fig0}.  Correctness of this automaton is easily proved
by induction on the length of the input, using the fact that
state $0$ corresponds to ``no carry'' and state $1$ corresponds
to ``carry expected''.
\begin{figure}[H]
\begin{center}
\includegraphics[width=4in]{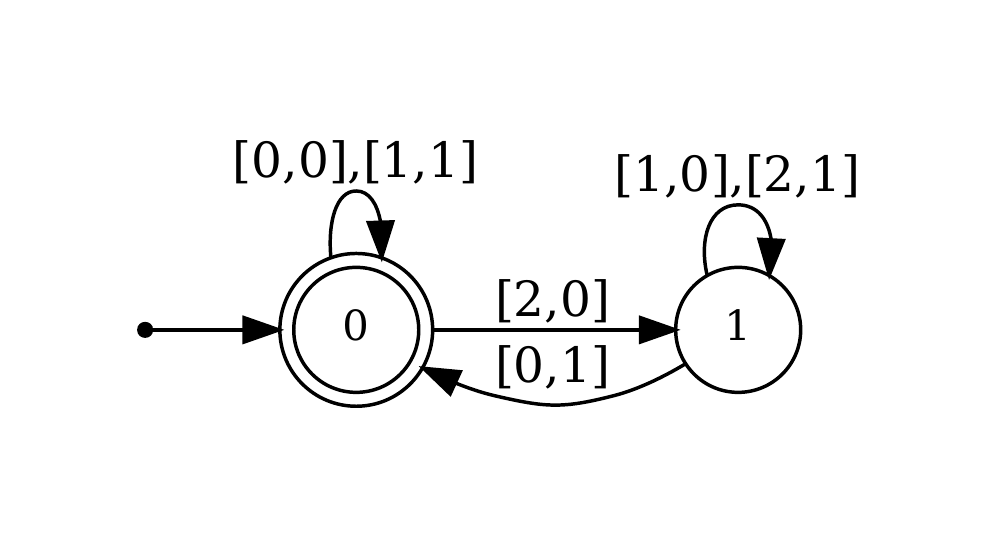}
\end{center}
\caption{Normalizer for base-$2$ expansions.}
\label{fig0}
\end{figure}

The final piece is the observation that the first $2^i$ bits of
$\bf t$ are just $\mu^i(0)$, and the first $2^i$ bits of
$\overline{\bf t}$ are $\mu^i(1)$.   Since a $2$-state automaton can compute
the $n$'th bit of $\bf t$ (or $\overline{\bf t}$), we can 
determine the appropriate bit.

Exactly the same idea works for the infinite Restivo words, except
now the code is an infinite word, so we need to use a B\"uchi automaton
in order to process it correctly.

The finite Salemi words are only slightly more complicated.
Here we use the (easily-verified) fact that 
$$ \mu^{t-1} g(d_{t-1}) \cdots
\mu^1 (g(d_1)) g(d_0) = w^R,$$ where
$$w =
\begin{cases}
g(d_0) \mu(\overline{g(d_1)}) \mu^2(g(d_2)) \mu^3(\overline{g(d_3)}) \cdots
        \mu^{t-1} (g(d_{t-1})), & \text{if $t$ odd}; \\
g(d_0) \mu(\overline{g(d_1)}) \mu^2(g(d_2)) \mu^3(\overline{g(d_3)}) \cdots
        \mu^{t-1} (\overline{g(d_{t-1})}), & \text{if $t$ even}.
\end{cases} $$
On input $n$, we use the lengths of the finite words
$g(c_0) \mu(g(c_1)) \mu^2(g(c_2)) \cdots
\mu^t(g(c_t))$ and $ \mu^{t-1} g(d_{t-1}) \cdots
\mu^1 (g(d_1)) g(d_0)$ to decide where the $n$'th symbol lies,
and then appeal to the lookup automaton for $R(c_0 \cdots c_t)$,
or its modification for $w$,, to compute the appropriate bit.

This completes our sketch of the decision procedure.
\end{proof}

For an infinite word $\bf x$, we can write a first-order
formulas asserting that $\bf x$ has an overlap (resp., has a $p/q$-power),
as follows:
\begin{align*}
\exists i,n \ (n\geq 1) \ \wedge\ \forall t \ (t\leq n) \implies 
{\bf x}[i+t]={\bf x}[i+t+n] \\
\exists i,n \ (n\geq 1) \ \wedge\ \forall t \ (qt<(p-q)n) \implies 
{\bf x}[i+t]={\bf x}[i+t+n] .
\end{align*}
Here $p$ and $q$ are positive integer constants and an expression
like $qt$ is shorthand for $\overbrace{t+t+\cdots+t}^{q\ \rm times}$.

So, incorporating these two formulas into larger first-order logical formulas
asserting that a given code specifies an overlap-free word (or
$\alpha$-free word
for rational or extended rational
$\alpha$ with $2 < \alpha \leq 7/3$), we immediately
get the following corollary:
\begin{corollary}
The first-order theory of the overlap-free words (or more generally,
$\alpha$-free words for rational or extended rational $\alpha$ with  $2 < \alpha \leq 7/3$), is decidable.
\end{corollary}

\section{Implementation}

We implemented part of the decision procedure discussed in
Section~\ref{decision} using {\tt Walnut}, a theorem-prover
originally designed by Hamoon Mousavi \cite{Mousavi:2016}.

The main part we implemented was for the finite Restivo words.
This allows us to solve many (but not all) questions about
infinite overlap-free words.   The limitation is because
{\tt Walnut} is based on ordinary finite automata and not
B\"uchi automata.

To implement our decision procedure in
{\tt Walnut}, we represent encodings as strings
over the alphabet $\{1,2,3,4,5\}$.   Since the encoded binary string
might need more binary digits to specify a position within it than
the number of symbols in the encoding, we also allow an arbitrary number
of trailing zeros in a code.

All numbers are represented in base $2$, starting with the 
{\it least significant digit}.

Our {\tt Walnut} solution needs various subautomata, as follows.
Most of these are deterministic finite automata (DFA), with the
exception of {\tt CODE} and {\tt LOOK}, which are
DFAO's.

\begin{itemize}

\item {\tt power2}:  one argument $n$. True if $n$ is a power of $2$ and
$0$ otherwise.

\item {\tt adjacent}:  two arguments $m, n$.  True if $m = 2^i$, $n = 2^{i-1}$
for some $i \geq 1$, or if $m = 1$ and $n = 0$.

\item {\tt hmorph}:  two arguments $c,y$.   True if $y$ represents applying
$h$ to the code specified by $c$.

\item {\tt validcode}:  one argument $c$.  True if $c$ represents a valid
code, that is, a word in $\{ 1,2,3,4,5 \}^*$ followed by $0$'s.

\item {\tt length}:   two arguments $c,n$.   True if $n$ is the length
of the binary string encoded by the codes $c$.

\item {\tt prefix}:   three arguments $a,b,c$.   Both $b,c$ are
are extended binary representations, while $a$ is either $0$ or
a power of $2$ in ordinary binary representation.  The result is true
if the word
$c$ equals $b$ copied digit-by-digit, up to and including the position
specified by the single $1$ in $a$, and $0$'s thereafter.   

\item {\tt CODE}:   a DFAO, two arguments $c, n$.   Returns the
code in $\{1,2,3,4,5\}$
corresponding to the digit specified by $n$, a power of $2$.

\item {\tt look1}:  two arguments $c, n$.  True if 
$R(c_0 c_1 \cdots c_{t-1})[n] = 1$ and $0$ otherwise (which
includes the case where the index $n$ is out of range).

\item {\tt look2}:  two arguments $c, n$.  True if the code
$c$ is invalid (for example, because it has interior $0$'s) or
the index $n$ is out of range.

\item {\tt LOOK}:  a DFAO, two arguments $c, n$.  Returns
$R(c_0 c_1 \cdots c_{t-1})[n]$ if the index is in range, 
and $2$ otherwise.  Obtained by combining the DFA's for
{\tt look1} and {\tt look2}.

\end{itemize}

Here is the {\tt Walnut} code for these.  A brief reminder
of {\tt Walnut}'s syntax may be necessary.   
\begin{itemize}
\item {\tt A} and {\tt E} represent the universal and existential
quantifiers, respectively.
\item {\tt lsd\_k} tells {\tt Walnut} to interpret numbers
in base-$k$, using least-significant-digit first representation.
\item {\tt |} is logical OR, {\tt \&} is logical AND,
{\tt =>} is logical implication, {\tt \char'176} is logical NOT.
\item {\tt reg} defines a regular expression.
\item {\tt def} defines an automaton accepting the representation
of free variables making the formula true.
\end{itemize}

\begin{verbatim}
reg power2 lsd_2 "0*10*":
def adjacent "?lsd_2 ($power2(m) & $power2(n) & m=2*n) | (m=1 & n=0)":
reg hmorph lsd_6 lsd_3 "([1,0]|[2,1]|[3,1]|[4,2]|[5,2])*[0,0]*":
reg validcode lsd_6 "(1|2|3|4|5)*0*":
reg prefix lsd_2 lsd_3 lsd_3 "(([0,0,0]|[0,1,0]|[0,2,0])*)|
   (([0,0,0]|[0,1,1]|[0,2,2])*)([1,0,0]|[1,1,1]|[1,2,2])
   ([0,0,0]|[0,1,0]|[0,2,0])*":
def length "?lsd_2 El $hmorph(?lsd_6 c,?lsd_3 l) & 
   $normalize(?lsd_3 l,?lsd_2 n)":
\end{verbatim}

In order to construct the automaton {\tt look1}, which is the most
complicated part of our construction, we use the following auxiliary
variables:
\begin{itemize}
\item $p$, the power of $2$ that corresponds to the particular
$\mu^i (g(c_i))$ block that the $n$'th bit falls in.
\item $q= \lfloor p/2 \rfloor$.
\item $l$, a number in extended binary representing the lengths of the
strings represented by the codes $c$.
\item $g$, a number in extended binary where we have cancelled from $l$
the bits corresponding to higher powers of $2$ than $p$.
\item $h$, a number in extended binary where we have cancelled from $l$
the bits corresponding to higher powers of $2$ than $q$.
\item $r$, a base-$2$ index giving the start of the block after which
$n$ appears.
\item $s$, a base-$2$ index giving the start of the block where $n$
appears.
\item $x$, the relative position inside the appropriate block corresponding
to the bit $n$.
\end{itemize}

Once these are ``guessed'' with an existential quantifier, we verify
them with the appropriate automata and then compute the appropriate
bit depending on the particular $c_i$, as follows:

\begin{verbatim}
def look1 "?lsd_2 Ep,q,l,g,h,r,s,x $validcode(?lsd_6 c) & $adjacent(p,q) &
   $hmorph(?lsd_6 c,?lsd_3 l) & $prefix(?lsd_2 p,?lsd_3 l,?lsd_3 g) &
   $prefix(?lsd_2 q,?lsd_3 l,?lsd_3 h) & $normalize(?lsd_3 g,?lsd_2 r) &
   $normalize(?lsd_3 h,?lsd_2 s) & n>=s & n<r & x+s=n &
   ((CODE[?lsd_2 p][?lsd_6 c]=@2 & T[x]=@1)
   |(CODE[?lsd_2 p][?lsd_6 c]=@3 & T[x]=@0)
   |(CODE[?lsd_2 p][?lsd_6 c]=@4 & x<p & T[x]=@1)
   |(CODE[?lsd_2 p][?lsd_6 c]=@4 & x>=p & T[x-p]=@1)
   |(CODE[?lsd_2 p][?lsd_6 c]=@5 & x<p & T[x]=@0)
   |(CODE[?lsd_2 p][?lsd_6 c]=@5 & x>=p & T[x-p]=@0))":
def look2 "?lsd_2 (~$validcode(?lsd_6 c)) | (El $length(?lsd_6 c,?lsd_2 l) 
   & n>=l):
combine LOOK look1=1 look2=2:
\end{verbatim}
The resulting DFAO, {\tt LOOK}, has 17 states.  We do not display it
here because its transition diagram is too complicated.

\section{Applications}

\subsection{Overlap-free words}

    We can now use this DFAO to obtain a number of
results.   First, let us find an automaton recognizing all finite
words $c_0 \cdots c_{t-1}$ such that $R(c_0 \cdots c_{t-1})$
is overlap-free.   This is done as follows:
\begin{verbatim}
def hasover "?lsd_2 At (t<=n) => LOOK[?lsd_6 c][i+t]=LOOK[?lsd_6 c][i+n+t]":
def ovlf "?lsd_2 $validcode(?lsd_6 c) & ~Ei,n,l $length(?lsd_6 c,?lsd_2 l)
   & n>=1 & i+2*n<l & $hasover(?lsd_6 c,?lsd_2 i,?lsd_2 n)":
reg good lsd_6 "(1|2|3|4|5)*":
def ovlfg "?lsd_6 $good(c) & $ovlf(c)":
\end{verbatim}
The resulting automaton is depicted in Figure~\ref{fig2}.
\begin{figure}[H]
\begin{center}
\includegraphics[width=6.5in]{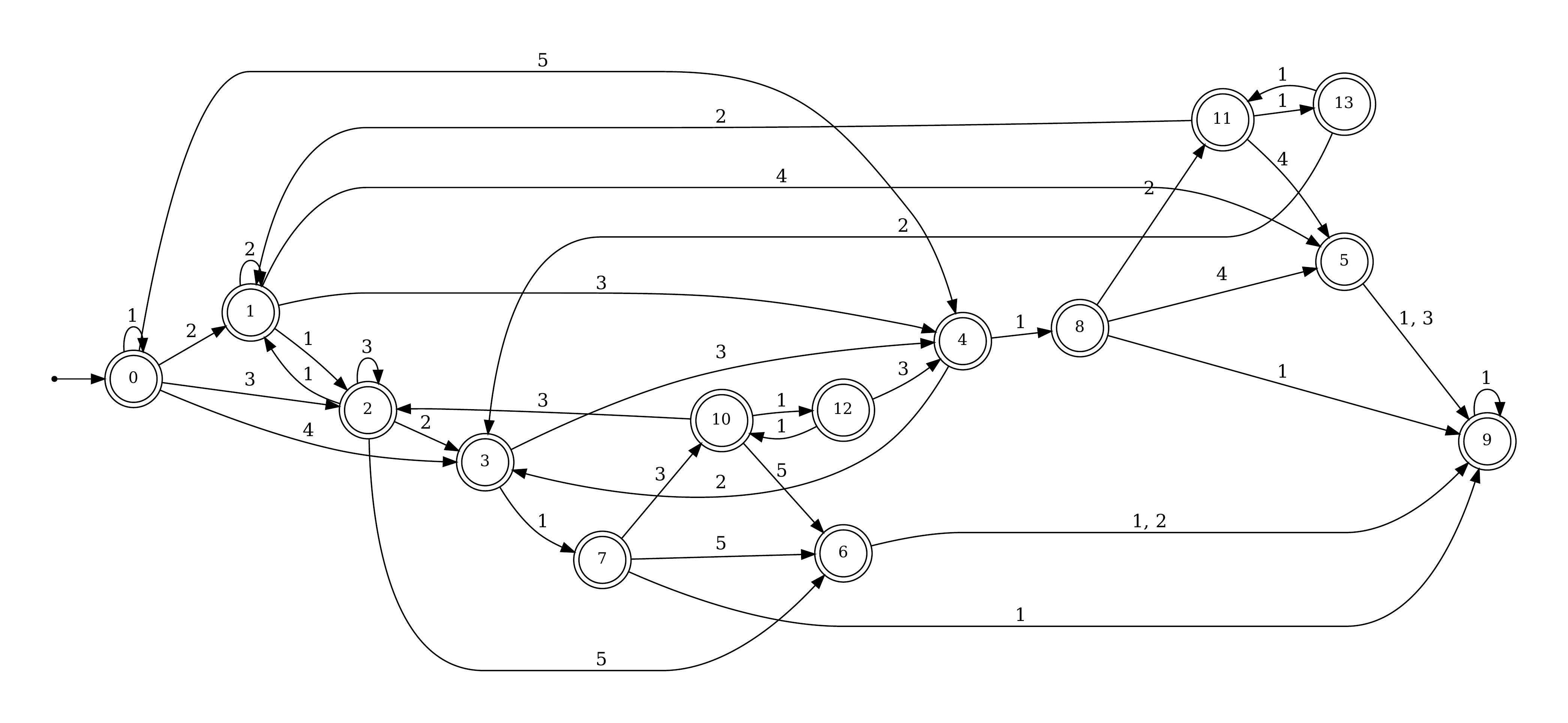}
\end{center}
\caption{Codes for overlap-free sequences.}
\label{fig2}
\end{figure}

This automaton essentially accepts all infinite strings $c_0 c_1 c_2 \cdots$
such that $R(c_0 c_1\ldots)$ is overlap-free.  However, there are
some subtleties that arise in interpreting it, due to the nature of our
encoding.    We describe them now.

When we compare the automaton in Figure~\ref{fig2} to that in
\cite{Shallit:2011a}, we see the following differences.  First,
the codes are different, and are related as follows:
\begin{table}[H]
\begin{center}
\begin{tabular}{c|c|c}
encoded word & old code & new code \\
\hline
$\epsilon$ & 0 & 1 \\
0 & 1 & 2 \\
1 & 3 & 3 \\
00 & 2 & 4 \\
11 & 4 & 5
\end{tabular}
\end{center}
\end{table}

Second, the names of states are different, and are related as follows:
\begin{table}[H]
\begin{center}
\begin{tabular}{c|c}
old state & new state \\
\hline
A & 0 \\
B & 1 \\
C & 3 \\
D & 2 \\
E & 4 \\
F & 7 \\
G & 10 \\
H & 12 \\
I & 8 \\
J & 11 \\
K & 13
\end{tabular}
\end{center}
\end{table}

Notice that the automaton in Figure~\ref{fig2} has three
additional states, numbered 5,6,9, that do not appear
in the automaton given in \cite{Shallit:2011a}.
The explanation for this is as follows:  the only accepting
paths from these states end in an infinite tail of $1$'s.
These paths can only correspond to either a suffix of $\bf t$ or
$\overline{\bf t}$, and in all cases the resulting
words have overlaps.  Therefore we can delete these states
5,6,9 from Figure~\ref{fig2} and obtain the automaton
given in \cite{Shallit:2011a}.

We now use our automaton for overlap-free words to prove a result,
about the lexicographically least overlap-free infinite word, previously
proved in \cite{Allouche&Currie&Shallit:1998}.
\begin{theorem}
The lexicographically least overlap-free infinite word is
$001001 \overline{\bf t}$.
\end{theorem}

\begin{proof}
We create a {\tt Walnut} formula that recognizes all finite code strings $c$
with the property that the overlap-free word $w$ specified by $c$ is
lexicographically $\leq$ all overlap-free words $w'$ with $|w'| \geq |w|$.
This can be done as follows:
\begin{verbatim}
reg good lsd_6 "(1|2|3|4|5)*":
def agrees "?lsd_2 At (t<b) => LOOK[?lsd_6 c1][t]=LOOK[?lsd_6 c2][t]":
# inputs (b,c1,c2)
# does the word specified by c1 agree with that specified by c2 
# on positions 0 through b-1?

def ispref "?lsd_2 El1,l2 $length(?lsd_6 c1,?lsd_2 l1) & 
   $length(?lsd_6 c2,?lsd_2 l2) & l1<=l2 & 
   $agrees(l1,?lsd_6 c1, ?lsd_6 c2)":
# code c1, c2
# yes if word coded by c1 is a prefix of that coded by c2

def lexlt "?lsd_2 El1,l2,m,i $length(?lsd_6 c1,?lsd_2 l1) &
   $length(?lsd_6 c2,?lsd_2 l2) & $min(l1,l2,m) & i<m & 
   $agrees(i,?lsd_6 c1, ?lsd_6 c2) & 
   LOOK[?lsd_6 c1][?lsd_2 i]<LOOK[?lsd_6 c2][?lsd_2 i]":

def lexlte "?lsd_6 $ispref(c1,c2) | $lexlt(c1,c2)":

def lexleast "?lsd_2 $good(c1) & $validcode(?lsd_6 c1) & $ovlf(?lsd_6 c1)
   & Ac2,l1,l2 ($validcode(?lsd_6 c2) & $ovlf(?lsd_6 c2) 
   & $length(?lsd_6 c2,?lsd_2 l2) & $length(?lsd_6 c1,?lsd_2 l1) & l1<=l2)
   =>  $lexlte(c1,c2)":
\end{verbatim}

The resulting automaton is depicted in Figure~\ref{fig11}.  This
was a rather big computation in {\tt Walnut}; the automaton for
{\tt agrees} has 122 states, and required 120G of RAM and
87762417 ms to compute.  The largest intermediate automaton 
had 3534633 states.
\begin{figure}[H]
\begin{center}
\includegraphics[width=5.5in]{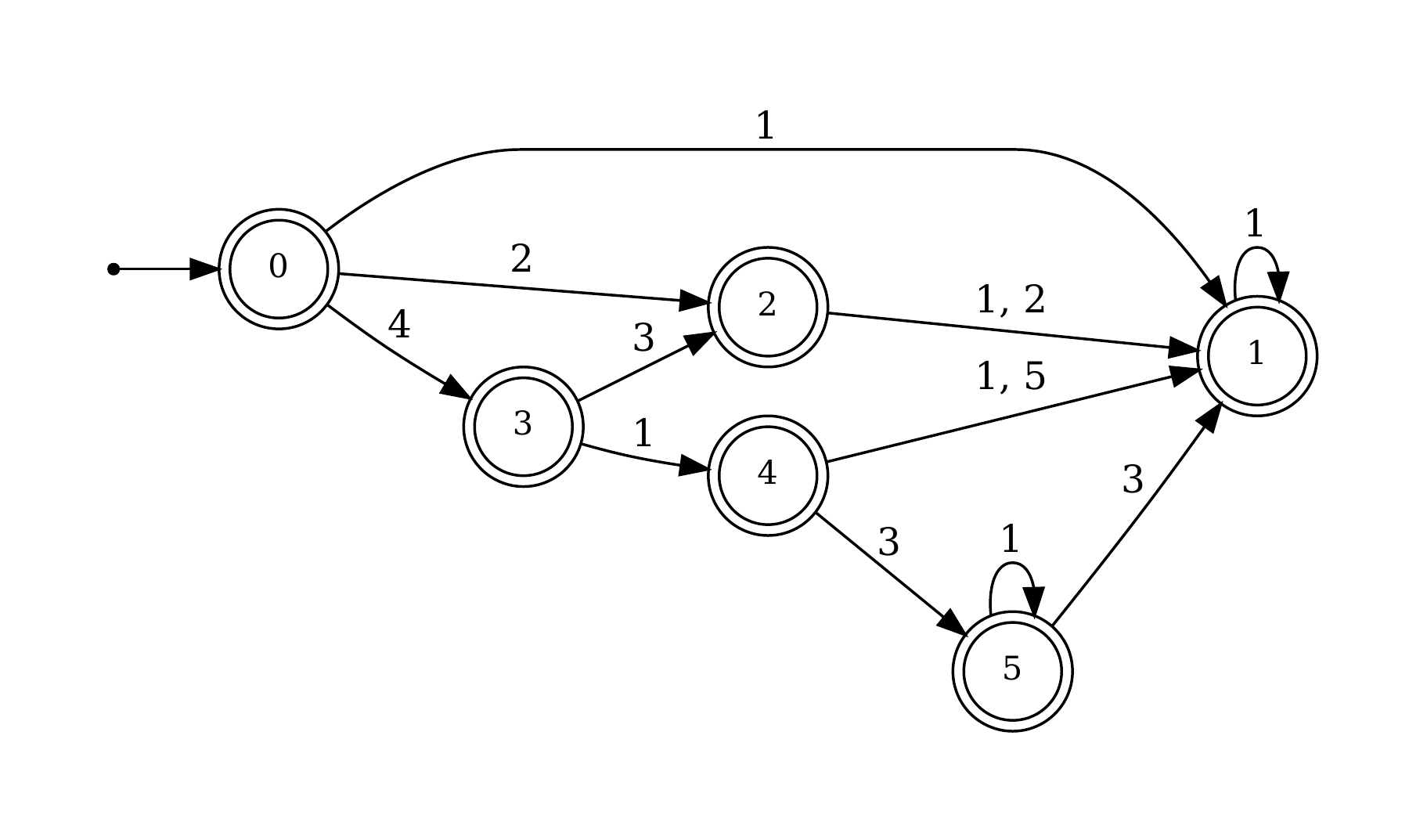}
\end{center}
\caption{Codes for lexicographically smallest words.}
\label{fig11}
\end{figure}
By inspection of this automaton, we see that the only arbitrarily long
accepting path that does not end in $1$'s is $4131^*3$.
This corresponds to the word $001001\overline{\bf t}$.
\end{proof}

\begin{remark}
Using our technique, we can also prove that the same word
$001001\overline{\bf t}$ is the lexicographically least
$7/3$-power-free word, and
and hence it is lexicographically least for
all $\alpha$-power-free words with $2< \alpha \leq 7/3$.
\end{remark}

Now we turn to the following theorem
from \cite{Brown&Rampersad&Shallit&Vasiga:2006}:  
\begin{theorem}
Take the Thue-Morse word $\bf t$ and flip any finite nonzero number
of bits, sending $0$ to $1$ and vice versa.  
Then the resulting word has an overlap.
\end{theorem}

At first glance this theorem does not seem susceptible to our technique,
because specifying an arbitrary finite set of positions to change
requires second-order logic.   But we can still prove it!  Instead
of quantifying over all finite sets of positions to change, we instead
quantify over all infinite overlap-free words, and ask for which
codes $c_0 c_1 c_2 \cdots$ the specified word agrees with Thue-Morse
on an infinite suffix.

If we had implemented our decision procedure for infinite Restivo words
using B\"uchi automata instead of ordinary finite automata, this would
be easy to translate into a first-order logical formula.  However, the
fact that our implementation can only deal with finite codes
$c_0 c_1 \cdots c_t$ makes it somewhat harder.

\begin{proof}
Instead, we use the following idea:  we design an automaton to accept
all finite codes $c_0 \cdots c_t$ with the property that there exists
arbitrarily long finite codes $d_0 \cdots d_s$ such that
\begin{itemize}
\item $c_0 \cdots c_t$ is a prefix of $d_0 \cdots d_s$;
\item $w = R(d_0 \cdots d_s)$ is overlap-free;
\item $|R(c_0 \cdots c_t)| = l$;
\item $w$ agrees with ${\bf t}$ on the positions from index $l$ to
index $|w|-1$.
\end{itemize}

This is done with the following {\tt Walnut} code:
\begin{verbatim}
reg prefixc lsd_6 lsd_6 "([1,1]|[2,2]|[3,3]|[4,4]|[5,5])*
([0,1]|[0,2]|[0,3]|[0,4]|[0,5])*[0,0]*":
reg lastnzcode lsd_6 lsd_2 "([1,0]|[2,0]|[3,0]|[4,0]|[5,0])*
   ([1,1]|[2,1]|[3,1]|[4,1]|[5,1])[0,0]*":
def tmagree "?lsd_2 El $length(?lsd_6 c,?lsd_2 l) & 
   At (t>=n & t<l) => LOOK[?lsd_6 c][t] = T[t]":
def changebits "?lsd_2 $good(?lsd_6 c) & El $length(?lsd_6 c,?lsd_6 l) &
   Az Ed,y $prefixc(?lsd_6 c,?lsd_6 d) & $length(?lsd_6 d,?lsd_2 y)
   & y>=z & $tmagree(?lsd_6 d,?lsd_2 l) & $ovlf(?lsd_6 d)":
\end{verbatim}
The resulting automaton only accepts $1^*$, so there are no
such codes except that specifying the Thue-Morse sequence.
\end{proof}

\subsection{${7\over 3}$-power-free words}

We now apply the method to re-derive the automaton given in
\cite{Rampersad&Shallit&Shur:2011} for ${7\over 3}$-power-free words.

\begin{verbatim}
def avoid73 "?lsd_2 $validcode(?lsd_6 c) & ~Ei,n,l 
   $length(?lsd_6 c,?lsd_6 l) & n>=1 & i+(7*n)/3<l & 
   At (3*t<4*n) => LOOK[?lsd_6 c][i+t]=LOOK[?lsd_6 c][i+n+t]":
def avoid73g "?lsd_6 $good(c) & $avoid73(c)":
\end{verbatim}

\begin{figure}[H]
\begin{center}
\includegraphics[width=6.5in]{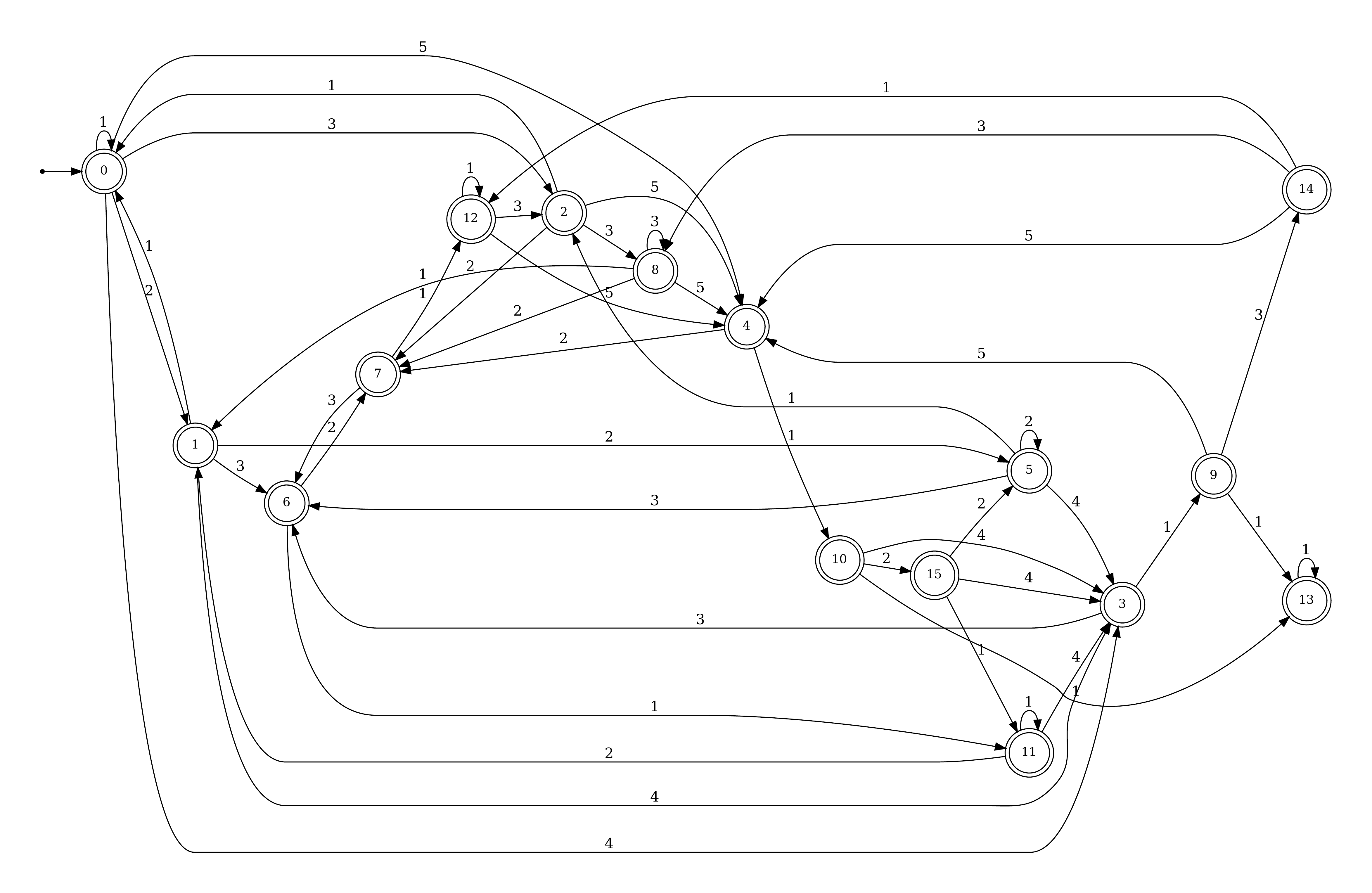}
\end{center}
\caption{Codes for ${7\over 3}$-power-free sequences.}
\label{fig3}
\end{figure}

This obtains, in a purely mechanical fashion, the automaton in Figure
2 of \cite{Rampersad&Shallit&Shur:2011} that was previously
constructed using a rather tedious examination of cases.  The
relationship between the old version in that paper and the new version
given here is summarized in Table~\ref{tab401}:
\begin{table}[H]
\begin{center}
\begin{tabular}{c|c}
old state & new state \\
\hline
$\epsilon$ & 0 \\
1 & 1 \\
2 & 3 \\
3 & 2 \\
4 & 4 \\
11 & 5 \\
13 & 6 \\
31 & 7 \\
33 & 8 \\
20 & 9 \\
40 & 10 \\
130 & 11 \\
310 & 12 \\
203 & 14 \\
401 & 15 
\end{tabular}
\end{center}
\label{tab401}
\end{table}
Once again there is a state, state 13, that appears in Figure~\ref{fig3}
but not in the paper \cite{Rampersad&Shallit&Shur:2011}.  Again, this
is because the only accepting path reachable from this state consists
of an infinite tail of $1$'s, which does not result in a ${7 \over 3}$-power-free word.

As an application, let us reprove a result from
\cite{Currie&Rampersad&Shallit:2006}:
\begin{theorem}
There exist uncountably many
infinite ${7\over3}$-power-free binary words, each containing arbitrarily
large overlaps.
\end{theorem}

\begin{proof}
We claim that every code in
$212 \{12,1112\}^i $
corresponds to a ${7\over 3}$-power-free word with overlaps of $i$
different lengths.
The automaton in Figure~\ref{fig3} clearly accepts every word in
$212\{12,1112\}^*$, so the words are
${7\over3}$-power-free.  To check the property of containing arbitrarily
large overlaps, we create an automaton that recognizes, in parallel,
those codes in $(211^*)^*2$, together with the lengths of
overlaps that occur in the resulting word.
\begin{verbatim}
reg two1 lsd_6 "(21(1*))*20*":
def large_overl "?lsd_2 El,i $length(?lsd_6 c,?lsd_2 l) & $two1(?lsd_6 c)
   & n>=1 & i+2*n<l & $hasover(?lsd_6 c,?lsd_2 i,?lsd_2 n)":
\end{verbatim}

\begin{figure}[H]
\begin{center}
\includegraphics[width=6.5in]{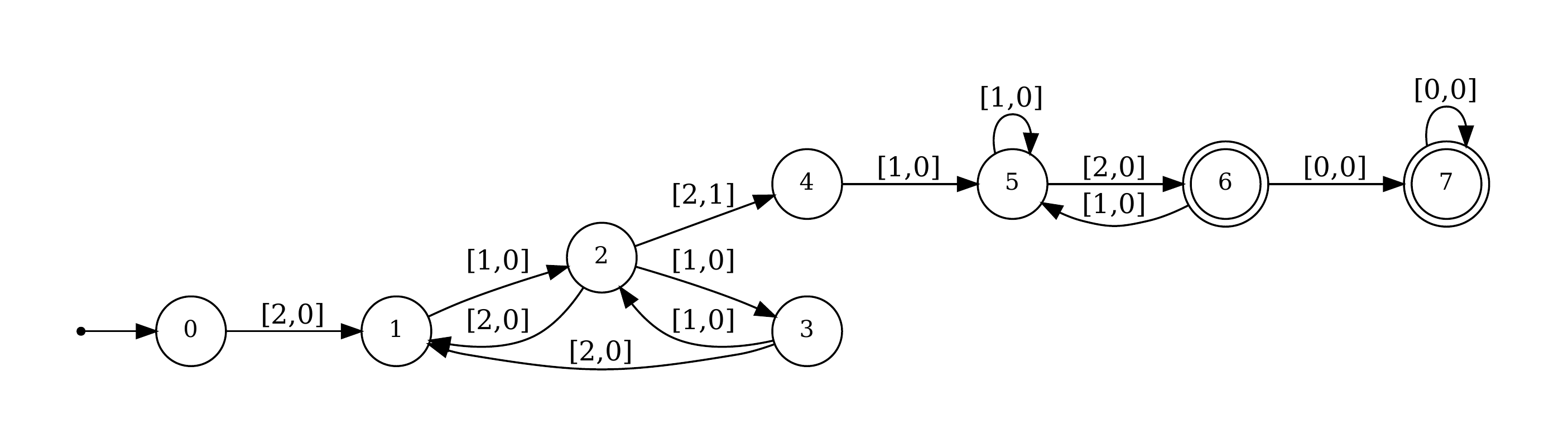}
\end{center}
\caption{Large overlaps in a ${7\over3}$-power-free word.}
\label{fig5}
\end{figure}
Inspection of the automaton in Figure~\ref{fig5} proves the claim.
Hence every word coded by $212 \{12,1112\}^\omega$ has overlaps of infinitely
many different lengths.
\end{proof}

\subsection{New results}

We can use the framework so far to prove a number of new results
about overlap-free and Restivo words.

For example, it is an easy consequence of the Restivo-Salemi theorem
that every infinite overlap-free binary word contains arbitrarily large
squares.   We can prove this and more in a quantitative sense.
\begin{theorem}
Every finite overlap-free word of length $l>7$ contains a square
of order $\geq l/6$.  Furthermore, the bound is best possible,
in the sense that there are arbitrarily large
overlap-free words for which the largest square is of order exactly $l/6$.
\end{theorem}

\begin{proof}
We can check the first claim with {\tt Walnut} as follows:
\begin{verbatim}
def has_square "?lsd_2 At (t<n) => LOOK[?lsd_6 c][i+t]=LOOK[?lsd_6 c][i+t+n]":
eval squ "?lsd_2 Ac,l ($ovlf(?lsd_6 c) & $length(?lsd_6 c,?lsd_2 l) & l>7)
   => Ei,n i+2*n<=l & 6*n>=l & $has_square(?lsd_6 c,?lsd_2 i,?lsd_2 n)":
\end{verbatim}
For the second claim, we can actually determine all code sequences
for which the largest square is of order exactly $l/6$.
\begin{verbatim}
def squ3 "?lsd_2 Ei,n,l $ovlf(?lsd_6 c) & $length(?lsd_6 c,?lsd_2 l)
   & i+2*n<=l & 6*n=l & $has_square(?lsd_6 c, ?lsd_2 i, ?lsd_2 n)":
def squ3b "?lsd_2 Ai,n,l ($ovlf(?lsd_6 c) & $length(?lsd_6 c,?lsd_2 l)
   & i+2*n<=l & $has_square(?lsd_6 c, ?lsd_2 i, ?lsd_2 n)) => 6*n<=l":
def squ4g "?lsd_6 $good(c) & $squ3(c) & $squ3b(c)":
\end{verbatim}
The resulting automaton is depicted in Figure~\ref{fig7}.
\begin{figure}[H]
\begin{center}
\includegraphics[width=6.5in]{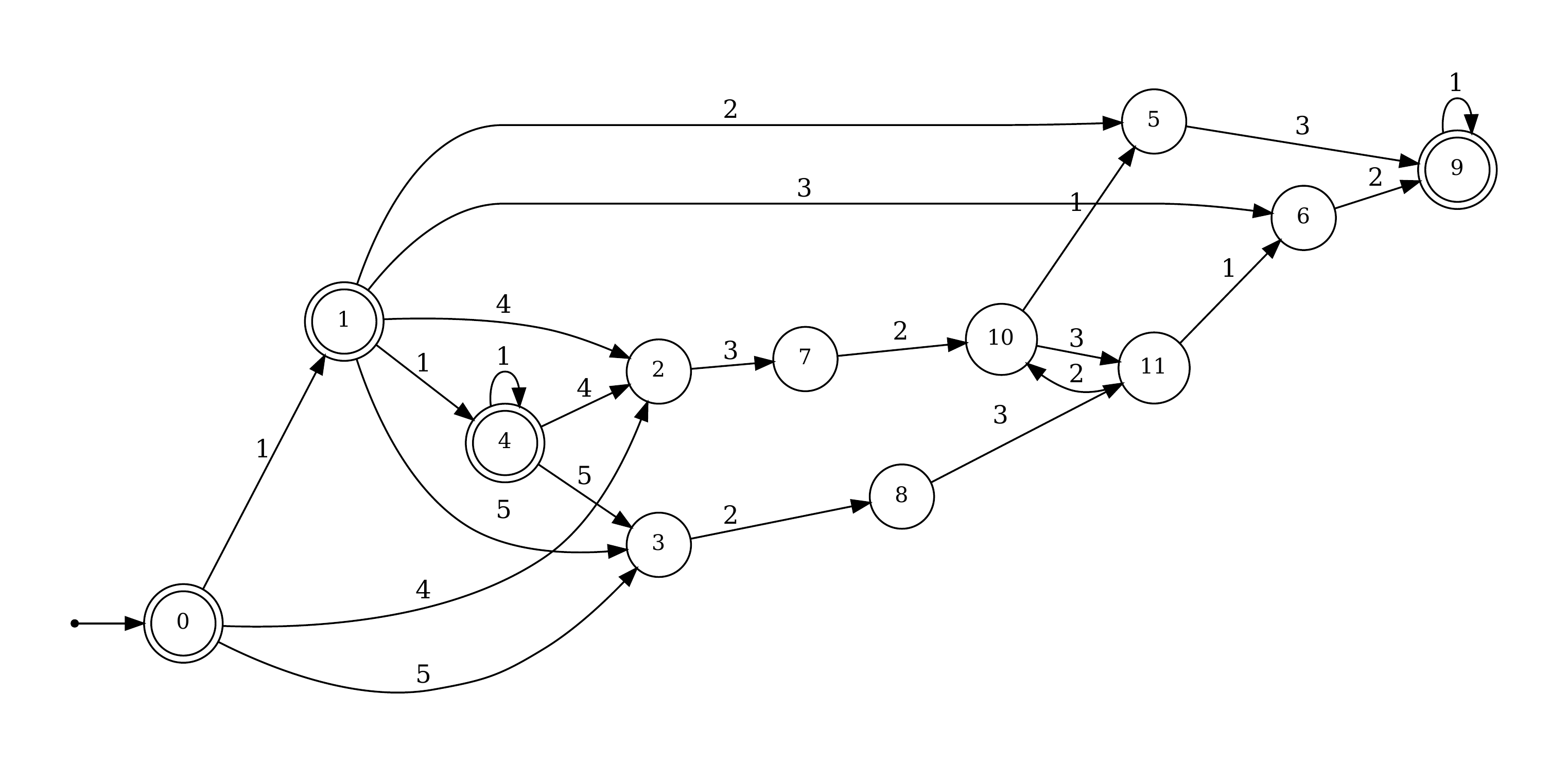}
\end{center}
\caption{Codes for length-$l$ overlap-free words with largest square
of order $l/6$.}
\label{fig7}
\end{figure}
In particular, the code sequence $4(32)^i13$ has length $6 \cdot 4^i$
and has largest square of order $4^i$.
\end{proof}

We can prove a similar, but weaker bound, for the larger class of
all Restivo words:
\begin{theorem}
Every finite Restivo word of length $l>8$ contains a square
of order $\geq (l+2)/7$.  Furthermore, the bound is best possible,
in the sense that there are arbitrarily large
overlap-free words for which the largest square is of order exactly $(l+2)/7$.
\end{theorem}

\begin{proof}
For the first statement we use
\begin{verbatim}
eval squaresin "?lsd_2 Ac,l ($validcode(?lsd_6 c) & $length(?lsd_6 c,?lsd_2 l)
   & l>8) => Ei,n i+2*n<=l & 7*n>=l+2 & $has_square(?lsd_6 c,?lsd_2 i,?lsd_2 n)":
\end{verbatim}
which evaluates to {\tt TRUE}.

For the second we construct an automaton accepting those code sequences
$c$ for which the largest square is of order exactly $(l+2)/7$.
\begin{verbatim}
def squr3 "?lsd_2 Ei,n,l $validcode(?lsd_6 c) & $length(?lsd_6 c,?lsd_2 l)
   & i+2*n<=l & 7*n=l+2 & $has_square(?lsd_6 c,?lsd_2 i,?lsd_2 n)":
def squr3b "?lsd_2 Ai,n,l ($validcode(?lsd_6 c) & $length(?lsd_6 c,?lsd_2 l)
   & i+2*n<=l & $has_square(?lsd_6 c,?lsd_2 i,?lsd_2 n)) => 7*n<=l+2":
def squr4g "?lsd_6 $good(c) & $squr3(c) & $squr3b(c)":
\end{verbatim}
The resulting automaton is depicted in Figure~\ref{fig9}.
\begin{figure}[H]
\begin{center}
\includegraphics[width=6.5in]{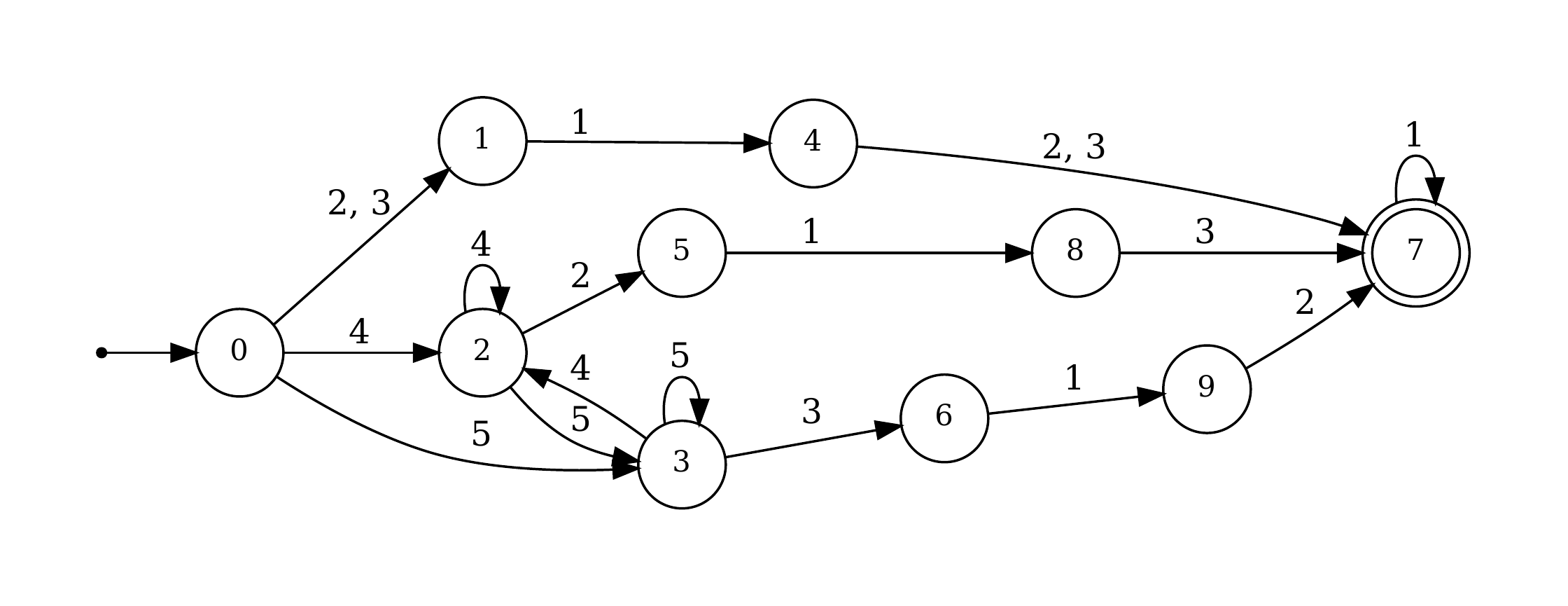}
\end{center}
\caption{Codes for length-$l$ Restivo words with largest square
of order $(l+2)/7$.}
\label{fig9}
\end{figure}
As you can see, code words of the form $5^i 312$ achieve the bound.
\end{proof}

As we have seen, not all code sequences result in overlap-free or ${7\over3}$-power-free words.   If we consider all code sequences, however, then
we can prove the following new result:
\begin{theorem}
\leavevmode
\begin{itemize}
\item[(a)]
Every (one-sided right) infinite word coded by a member
of $\{1,2,3,4,5\}^\omega$ is $4$th-power-free.
\item[(b)]
Furthermore, this bound is best possible, in the sense that
for each exponent $e < 4$ there is an infinite word
coded by a code in $\{1,2,3,4,5\}^\omega$
having a critical exponent $>e$.
\end{itemize}
\end{theorem}

\begin{proof}
\leavevmode
\begin{itemize}
\item[(a)]
We can check this with {\tt Walnut} as follows:
\begin{verbatim}
eval fourthr "?lsd_2 Ei,n,l,c $validcode(?lsd_6 c) & 
   $length(?lsd_6 c,?lsd_6 l) & n>=1 & i+3*n<=l & At (t<3*n)
   => LOOK[?lsd_6 c][i+t]=LOOK[?lsd_6 c][i+n+t]":
\end{verbatim}
This asserts the existence of a $4$th power, and returns {\tt FALSE}, so
no fourth power exists.

\item[(b)]
This requires a little more work.   What we do is show that
for all finite codes of length $t \geq 3$, there is a code
resulting in a word having a factor of length $2^t-1$ with
period $2^{t-2}$, and hence an exponent of $4(1 - 2^{-t})$.
\begin{verbatim}
reg lastnzcode lsd_6 lsd_2 "([1,0]|[2,0]|[3,0]|[4,0]|[5,0])*
   ([1,1]|[2,1]|[3,1]|[4,1]|[5,1])[0,0]*":
# last bit x with a nonzero code
# input is c,x
def maxexp "?lsd_2 Ex,l,i $lastnzcode(?lsd_6 c,?lsd_2 x) &
   $length(?lsd_6 c,?lsd_2 l) & i+2*x+2<=l+1 & At (t<3*x/2-1)
   => LOOK[?lsd_6 c][i+t]=LOOK[?lsd_6 c][i+t+x/2]":
\end{verbatim}
The resulting automaton is depicted in Figure~\ref{fig4}.
\begin{figure}[H]
\begin{center}
\includegraphics[width=5.5in]{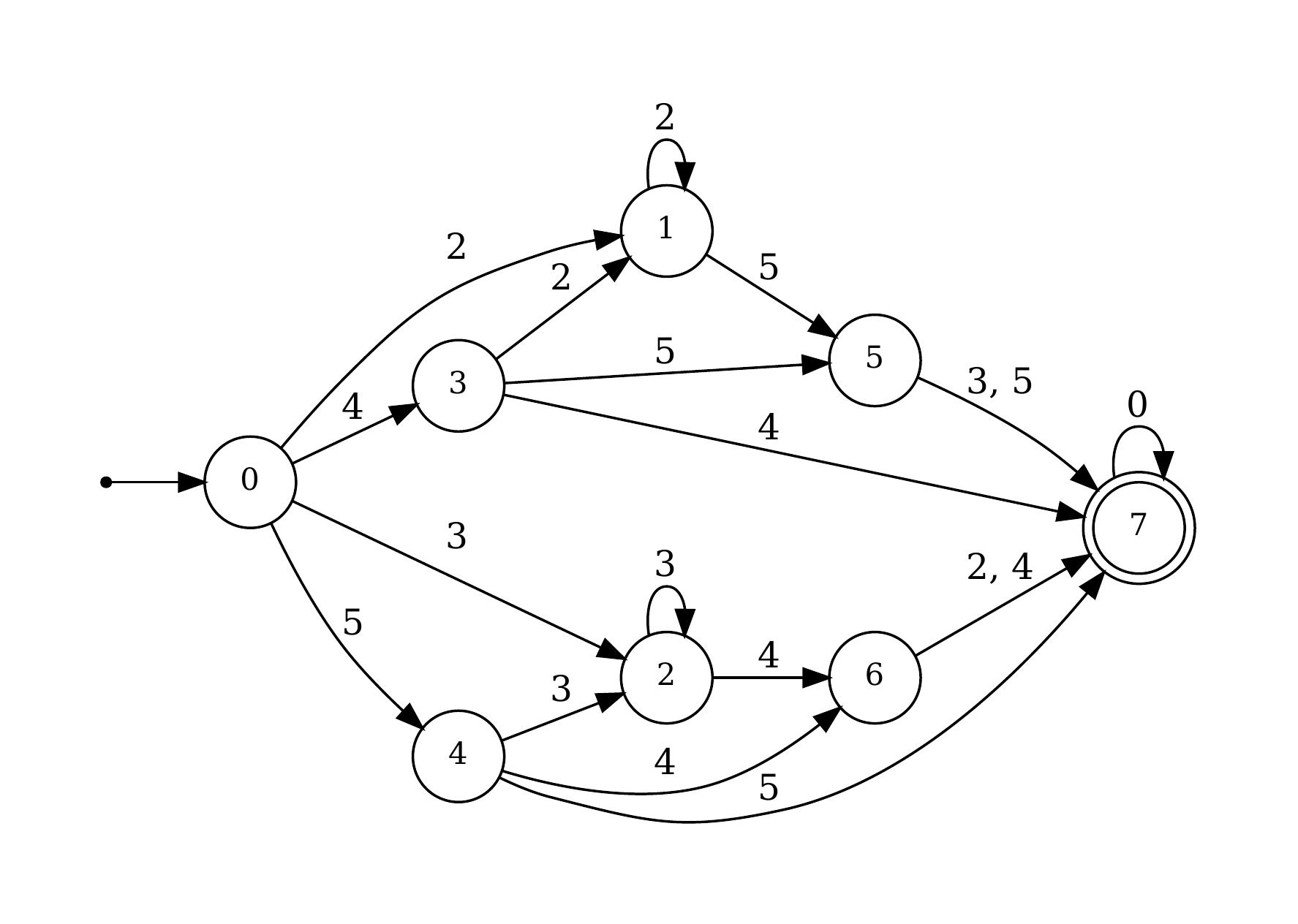}
\end{center}
\caption{Codes for words of critical exponent close to $4$.}
\label{fig4}
\end{figure}
From this, we see that the codes of length $t$ specifying
a word with critical exponent at least $4(1-2^{-t})$ are
$$ 2^{t-2}5\{3,5\}, 3^{t-2}4\{2,4\}, 42^{t-3}5\{3,5\}, 53^{t-3}4\{2,4\}.$$
\end{itemize}
\end{proof}

\section{Enumeration}
As discussed in several previous papers (e.g., \cite{Charlier&Rampersad&Shallit:2012,Du&Mousavi&Schaeffer&Shallit:2016}) the automaton-based technique
can also be used to {\it enumerate}, not simply decide, certain aspects of
sequences.   

Here we will use these ideas to enumerate the
``irreducibly extensible words'' of Kobayashi \cite{Kobayashi:1988}:
these are binary words $x$ such that there exists an infinite binary
word $\bf y$ such that $x{\bf y}$ is overlap-free.
For example, it is easily checked that
$010011001011010010$ is extendable,
but $010011001011010011$ is not (every extension by a word of length $7$
gives an overlap).
Denote
the number of such words as $E(n)$.  Table~\ref{en} gives the
first few values of this sequence.
\begin{table}[H]
\begin{center}
\begin{tabular}{c|cccccccccccccccc}
$n$ & 1 & 2 & 3 & 4 & 5 & 6 & 7 & 8 & 9 & 10 & 11 & 12 & 13 & 14 & 15 & 16 \\
\hline
$E(n)$ & 2& 4& 6& 10& 14& 18& 22& 26& 32& 36& 40& 44& 48& 52& 58& 64
\end{tabular}
\end{center}
\caption{First few values of $E(n)$.}
\label{en}
\end{table}
\noindent This is sequence \seqnum{A356959} in the
{\it On-Line Encyclopedia of Integer Sequences} \cite{Sloane}.

We can now obtain the following result of Kobayashi \cite{Kobayashi:1988}:
\begin{theorem}
$E(n) = \Theta(n^c)$, for $c \doteq 1.15501186367066470321$.
\end{theorem}

\begin{proof}
In order to carry out the enumeration, we need to create a first-order
logical formula asserting that $c$ is a code for an overlap-free sequence
of length at least $n$, and also that $c$ is lexicographically first
with this property
in some appropriate order.
The easiest lexicographic order results from interpreting
$c$ as a number in base $6$.  Then, counting the number of such codes
corresponding to each $n$ gives $E(n)$.  
We can carry this out with the following
{\tt Walnut} code:
\begin{verbatim}
def agrees "?lsd_2 At (t<b) => LOOK[?lsd_6 c1][t]=LOOK[?lsd_6 c2][t]":
def prefixequal "?lsd_2 El,m $length(?lsd_6 c1,?lsd_2 l) & 
   $length(?lsd_6 c2,?lsd_2 m) & l>=n & m>=n & 
   $agrees(?lsd_2 n,?lsd_6 c1, ?lsd_6 c2)":
def mincode "?lsd_2 El $ovlf(?lsd_6 c1) & $length(?lsd_6 c1, ?lsd_2 l) &
   l>=n & Ac2 ($prefixequal(?lsd_6 c1,?lsd_6 c2,?lsd_2 n) &
   $ovlf(?lsd_6 c2)) => (?lsd_6 c1<=c2)":
def minmat n "$mincode(?lsd_6 c,?lsd_2 n)":
\end{verbatim}
    Here {\tt Walnut} returns a so-called {\it linear representation\/}
for $E(n)$:  this consists of a row vector $v$, a matrix-valued
morphism $\gamma$, and a column vector $w$ such that
$E(n) = v \gamma(x) w$ if $x$ is a binary word with $[x]_2 = n$.
(For more about linear representations, see the 
book \cite{Berstel&Reutenauer:2011}.)  The {\it rank\/} of a linear
representation is the dimension of the vector $v$; in this case it is
$57$.    With this linear representation in hand, we can compute $E(n)$ very
rapidly even for large $n$.

The linear representation also can give us information about the
asymptotic behavior of $E(n)$.  To do so, it suffices to compute
the minimal polynomial of the matrix $\gamma(0)$ with a computer
algebra system such as {\tt Maple}; it is 
$X^4 (X^4-2X^3-X^2+2X-2)(X-1)^2(X+1)^2$.  Here the dominant zero
is that of $X^4-2X^3-X^2+2X-2$, and it is 
$$\zeta = {{1+ \sqrt{5 + 4 \sqrt{3}}}\over 2} \doteq 2.22686154846556164.$$
It follows that $E(2^n) \sim \alpha \cdot \zeta^n$ for some constant $\alpha$;
since $E$ is strictly increasing, it follows that $E(n) = \Theta(n^c)$ 
for $c = \log_2(\zeta) \doteq 1.15501186367066470321$.
\end{proof}

\section{Going further}

All the needed {\tt Walnut} code can be downloaded from the website
of the second author, \\
	\centerline{\url{https://cs.uwaterloo.ca/~shallit/walnut.html} \ .}

\medskip

In principle, one can extend this work to the Salemi words, and
we were able to construct the needed lookup automaton, which has 124 states.
However, so far we have been unable to use it to do much that is useful with
it, because
of the very large sizes of the intermediate automata (at least hundreds of
millions of states).  We leave this as a problem for future work.


\begin{thebibliography}{10}

\bibitem{Allouche&Currie&Shallit:1998}
J.-P. Allouche, J.~Currie, and J.~Shallit.
\newblock Extremal infinite overlap-free binary words.
\newblock {\em Electron. J. Combin.} {\bf 5} (1998), R27 (electronic).
\newblock \verb!www.combinatorics.org/ojs/index.php/eljc/article/view/v5i1r27!

\bibitem{Berstel:1995}
J.~Berstel.
\newblock {\em Axel {Thue's} Papers on Repetitions in Words: a Translation}.
\newblock Number~20 in Publications du Laboratoire de Combinatoire et
  d'Informatique {Math\'ematique}. Universit\'e du Qu\'ebec \`a Montr\'eal,
  February 1995.

\bibitem{Berstel&Reutenauer:2011}
J.~Berstel and C.~Reutenauer.
\newblock {\em Noncommutative Rational Series with Applications}, Vol. 137 of
  {\em Encyclopedia of Mathematics and Its Applications}.
\newblock Cambridge University Press, 2011.

\bibitem{Blondel&Cassaigne&Jungers:2009}
V.~D. Blondel, J.~Cassaigne, and R.~M. Jungers.
\newblock On the number of $\alpha$-power-free binary words for $2 < \alpha
  \leq 7/3$.
\newblock {\em Theoret. Comput. Sci.} {\bf 410} (2009), 2823–2833.

\bibitem{Brown&Rampersad&Shallit&Vasiga:2006}
S.~Brown, N.~Rampersad, J.~Shallit, and T.~Vasiga.
\newblock Squares and overlaps in the {Thue-Morse} sequence and some variants.
\newblock {\em RAIRO Inform. Th\'eor. App.} {\bf 40} (2006), 473--484.

\bibitem{Bruyere&Hansel&Michaux&Villemaire:1994}
V.~{Bruy\`ere}, G.~Hansel, C.~Michaux, and R.~Villemaire.
\newblock Logic and $p$-recognizable sets of integers.
\newblock {\em Bull. Belgian Math. Soc.} {\bf 1} (1994), 191--238.
\newblock Corrigendum, {\it Bull.\ Belg.\ Math.\ Soc.} {\bf 1} (1994), 577.

\bibitem{Buchi:1960}
J.~R. {B\"uchi}.
\newblock Weak secord-order arithmetic and finite automata.
\newblock {\em Z. Math. Logik Grundlagen Math.} {\bf 6} (1960), 66--92.
\newblock Reprinted in S. Mac Lane and D. Siefkes, eds., {\it The Collected
  Works of J. Richard {B\"uchi}}, Springer-Verlag, 1990, pp.\ 398--424.

\bibitem{Carpi:1993a}
A.~Carpi.
\newblock Overlap-free words and finite automata.
\newblock {\em Theoret. Comput. Sci.} {\bf 115} (1993),
  243--260.

\bibitem{Cassaigne:1993}
J.~Cassaigne.
\newblock Counting overlap-free binary words.
\newblock In P.~Enjalbert, A.~Finkel, and K.~W. Wagner, editors, {\em STACS
  93}, Vol. 665 of {\em Lecture Notes in Computer Science}, pp.  216--225.
  Springer-Verlag, 1993.

\bibitem{Charlier&Rampersad&Shallit:2012}
{\'E}.~Charlier, N.~Rampersad, and J.~Shallit.
\newblock Enumeration and decidable properties of automatic sequences.
\newblock {\em Internat. J. Found. Comp. Sci.} {\bf 23} (2012), 1035--1066.

\bibitem{Currie&Rampersad&Shallit:2006}
J.~Currie, N.~Rampersad, and J.~Shallit.
\newblock Binary words containing infinitely many overlaps.
\newblock {\em Electron. J. Combin.} {\bf 13} (2006), \#R82 (electronic).

\bibitem{Du&Mousavi&Schaeffer&Shallit:2016}
C.~F. Du, H.~Mousavi, L.~Schaeffer, and J.~Shallit.
\newblock Decision algorithms for {Fibonacci}-automatic words {III}:
  Enumeration and abelian properties.
\newblock {\em Internat. J. Found. Comp. Sci.} {\bf 27} (2016), 943--963.

\bibitem{Fici&Shallit:2022}
G.~Fici and J.~Shallit.
\newblock Properties of a class of {Toeplitz} words.
\newblock {\em Theoret. Comput. Sci.} {\bf 922} (2022), 1--12.

\bibitem{Fife:1980}
E.~D. Fife.
\newblock Binary sequences which contain no {$BBb$}.
\newblock {\em Trans. Amer. Math. Soc.} {\bf 261} (1980), 115--136.

\bibitem{Goc&Mousavi&Schaeffer&Shallit:2015}
D.~Go\v{c}, H.~Mousavi, L.~Schaeffer, and J.~Shallit.
\newblock A new approach to the paperfolding sequences.
\newblock In A.~Beckmann et~al., editor, {\em Computability in Europe, Cie
  2015}, Vol. 9136 of {\em Lecture Notes in Computer Science}, pp.  34--43.
  Springer-Verlag, 2015.

\bibitem{Hieryonymi&Ma&Oei&Schaeffer&Schulz&Shallit:2022}
P.~Hieronymi, D.~Ma, R.~Oei, L.~Schaeffer, C.~Schulz, and J.~Shallit.
\newblock Decidability for {Sturmian} words.
\newblock In F.~Manea and A.~Simpson, editors, {\em 30th EACSL Annual
  Conference on Computer Science Logic (CSL 2022)}, Leibniz International
  Proceedings in Informatics, pp.  24:1--24:23. Schloss
  Dagstuhl---Leibniz-Zentrum f{\"u}r Informatik, Dagstuhl Publishing, 2022.

\bibitem{Jungers&Protasov&Blondel:2009}
R.~M. Jungers, V.~Y. Protasov, and V.~D. Blondel.
\newblock Overlap-free words and spectra of matrices.
\newblock {\em Theoret. Comput. Sci.} {\bf 410} (2009), 3670--3684.

\bibitem{Karhumaki&Shallit:2004}
J.~Karhum{\"{a}}ki and J.~Shallit.
\newblock Polynomial versus exponential growth in repetition-free binary words.
\newblock {\em J. Combin. Theory. Ser. A} {\bf 105}(2) (2004), 335--347.

\bibitem{Kobayashi:1986}
Y.~Kobayashi.
\newblock Repetition-free words.
\newblock {\em Theoret. Comput. Sci.} {\bf 44} (1986), 175--197.

\bibitem{Kobayashi:1988}
Y.~Kobayashi.
\newblock Enumeration of irreducible binary words.
\newblock {\em Disc. Appl. Math.} {\bf 20} (1988), 221--232.

\bibitem{Mousavi:2016}
H.~Mousavi.
\newblock Automatic theorem proving in {Walnut}.
\newblock Preprint, available at \url{http://arxiv.org/abs/1603.06017}, 2016.

\bibitem{Rampersad:2007}
N.~Rampersad.
\newblock {\em Overlap-Free Words and Generalizations}.
\newblock PhD thesis, University of Waterloo, 2007.

\bibitem{Rampersad&Shallit&Shur:2011}
N.~Rampersad, J.~Shallit, and A.~Shur.
\newblock Fife's theorem for $(7/3)$-powers.
\newblock In P.~Ambroz, S.~Holub, and Z.~Masakova, editors, {\em WORDS 2011},
  Lecture Notes in Computer Science, pp.  189--198. Springer-Verlag, 2011.

\bibitem{Restivo&Salemi:1983}
A.~Restivo and S.~Salemi.
\newblock On weakly square free words.
\newblock {\em Bull. European Assoc. Theor. Comput. Sci.} , No.\ 21, (October
  1983), 49--56.

\bibitem{Restivo&Salemi:1985a}
A.~Restivo and S.~Salemi.
\newblock Overlap free words on two symbols.
\newblock In M.~Nivat and D.~Perrin, editors, {\em Automata on Infinite Words},
  Vol. 192 of {\em Lecture Notes in Computer Science}, pp.  198--206.
  Springer-Verlag, 1985.

\bibitem{Shallit:2011a}
J.~Shallit.
\newblock Fife's theorem revisited.
\newblock In G.~Mauri and A.~Leporati, editors, {\em DLT 2011}, Vol. 6795 of
  {\em Lecture Notes in Computer Science}, pp.  397--405. Springer-Verlag,
  2011.

\bibitem{Shallit:2022}
J.~Shallit.
\newblock {\em The Logical Approach to Automatic Sequences: Exploring
  Combinatorics on Words with {\tt Walnut}}.
\newblock Cambridge University Press, 2022.

\bibitem{Shur:2000}
A.~M. Shur.
\newblock The structure of the set of cube-free $\mathbb{Z}$-words in a
  two-letter alphabet ({Russian}).
\newblock {\em Izv. Ross. Akad. Nauk Ser. Mat.} {\bf 64} (2000), 201--224.
\newblock English translation in \emph{Izv. Math.} \textbf{64} (2000),
  847--871.

\bibitem{Sloane}
N.~J.~A. Sloane et~al.
\newblock The on-line encyclopedia of integer sequences, 2022.
\newblock Available at \url{https://oeis.org}.

\bibitem{Thue:1912}
A.~Thue.
\newblock {\"Uber} die gegenseitige {Lage} gleicher {Teile} gewisser
  {Zeichenreihen}.
\newblock {\em Norske vid. Selsk. Skr. Mat. Nat. Kl.} {\bf 1} (1912), 1--67.
\newblock Reprinted in {\it Selected Mathematical Papers of Axel Thue}, T.
  Nagell, editor, Universitetsforlaget, Oslo, 1977, pp.~413--478.

\end{thebibliography}
\end{document}